\documentclass{amsart}

\usepackage{epsfig}
\usepackage[all]{xy}
\usepackage{algorithmic}
\usepackage{algorithm}
\usepackage{enumerate}
\usepackage{framed}

\newtheorem{theorem}{Theorem}[section]
\newtheorem{cor}[theorem]{Corollary}
\newtheorem{lemma}[theorem]{Lemma}
\newtheorem{proposition}[theorem]{Proposition}

\theoremstyle{definition}
\newtheorem{definition}[theorem]{Definition}
\newtheorem{example}[theorem]{Example}

\theoremstyle{remark}
\newtheorem{remark}[theorem]{Remark}

\numberwithin{equation}{section}

\newcommand{\Cc}{\mathcal{C}}

\newcommand{\Z}{{\mathbb Z}}
\newcommand{\R}{{\mathbb R}}
\newcommand{\C}{{\mathbb C}}

\newcommand{\Q}{{\mathbb Q}}

\newcommand{\SL}{\operatorname{SL}}
\newcommand{\M}{\operatorname{M}}

\newcommand{\inte}{\operatorname{int}}
\newcommand{\ext}{\operatorname{ext}}
\newcommand{\Id}{\operatorname{Id}}

\title[Nonuniform Fuchsian Codes for Noisy Channels]{Nonuniform Fuchsian Codes for Noisy Channels}

\author[I. Blanco-Chac\'on]{Iv\'an Blanco-Chac\'on$^*$}
\thanks{$^{*}$Partially supported by  MTM2010-17389 (MICINN/ICMAT, Spain)}
\address{Aalto University, Department of Mathematics and Systems Analysis, P.O. Box 11100, FI-00076 AALTO, Helsinki, Finland.}
\email{ivan.blancochacon@aalto.fi}

\author[D. Rem\'on]{Dion\'is Rem\'on$^{**}$}
\thanks{$^{**}$Partially supported by MTM2012-33830 (MICINN/UB, Spain)}
\address{University of Barcelona, Faculty of Mathematics. Gran Via de les Corts Catalanes 585, 08007 Barcelona, Spain.}
\email{dremon@ub.edu}

\author[C. Hollanti]{Camilla Hollanti}
\address{Aalto University, Department of Mathematics and Systems Analysis, P.O. Box 11100, FI-00076 AALTO, Helsinki, Finland.}
\email{camilla.hollanti@aalto.fi}

\author[M. Alsina]{Montserrat Alsina$^{***}$}
\thanks{$^{***}$Partially supported by MTM2012-33830 (MICINN/UB, Spain)}
\address{Universitat Polit\`{e}cnica de Catalunya- BarcelonaTech, Dept. Applied Mathematics III - EPSEM, Av. Bases de Manresa 61-73,  08242  Manresa, Spain.}
\email{montserrat.alsina@upc.edu}



\begin{document}

\maketitle

\begin{abstract}We develop a new transmission scheme for additive white Gaussian noisy (AWGN) channels  based on Fuchsian groups from rational quaternion algebras. The structure of the proposed Fuchsian codes is nonlinear and nonuniform, hence conventional decoding methods based on linearity and symmetry do not apply. Previously, only brute force decoding methods with complexity that is linear in the code size exist for general nonuniform codes. However, the properly discontinuous character of the action of the Fuchsian groups on the complex upper half-plane translates into decoding complexity that is logarithmic in the code size via a recently introduced point reduction algorithm.
\end{abstract}


\section*{Introduction}
\label{introduction}
Fuchsian groups constructed from quaternion algebras arise in the study of Shimura curves \cite{shimura1967}, a rich theory with a large number of theoretical applications to various branches of number theory like Jacquet-Langlands correspondence or the proof of the  Shimura-Taniyama-Weil conjecture.  Shimura curves are also present in the theory of error-correcting codes \cite{elkies}. More recently, Fuchsian groups have made an appearance \cite{brazilian_franklin2,brazilian_IEEE,brazilian_COAM,brazilian_franklin} in the context of signal constellation design with potential applications in communications.

In this paper\footnote{A preliminary and partial version of this paper was presented at the International Workshop on Coding and Cryptography (WCC) 2013 \cite{WCC}.}, we will consider  a new family of \emph{Fuchsian codes}. The codes are obtained from unit groups of orders of quaternion algebras acting on the complex upper half-plane, in this way giving rise to complex points that can be used as codewords. Each of the above notions will be properly introduced in the sequel, but let us first concentrate on the general communication problem at hand.

Namely, as the underlying mathematical communication model, we will use the typical additive white Gaussian noise (AWGN) channel model \cite[Ch. 10]{IT}. The transmission process is described by the equation
\begin{equation}\label{AWGN}
y=x+w,
\end{equation}
where $y\in\C$ is the received signal, $x\in\C$ is the transmitted codeword drawn from a finite codebook $\Cc\subset\C$ (also referred to as a constellation), and $w$ is complex AWGN with zero mean and variance $\sigma^2/2$ per real and imaginary part.

Throughout this paper,  we denote by $\Re(z)$ and $\Im(z)$ the real and imaginary part of a complex number $z\in\mathbb{C}$, respectively. The complex absolute value, \emph{i.e.}, Euclidean norm  is denoted by $|z|=\sqrt{\Re(z)^2+\Im(z)^2}$, and the cardinality of a code $\Cc$ by $|\Cc|$. In spite of the slight abuse of notation there should not be any danger of confusion.

\subsection{Contributions, related work and organization}
Next, we summarize our main contributions and reflect our work to relevant earlier work related to Fuchsian groups in the context of communication applications. The main contributions of this paper are:
\begin{itemize}
\item[$\bullet$] We show how to explicitly build nonuniform signal constellations on the complex plane by using Fuchsian groups and M\"obius transformations. Non\-uni\-form signal constellations are included in the digital video broadcasting standard for next generation handheld (DVB-NGH) systems, and they are currently being considered for the future extension of terrestrial DVB with multiple antennas (DVB-T2 MIMO). This creates a great interest and need for nonuniform constellations.
\item[$\bullet$] We describe the whole encoding and decoding process of the proposed Fuchsian codes in full detail, assuming the AWGN communication setting.
\item[$\bullet$] Our construction method allows for decoding complexity which is logarithmic in the code size, enabled by the so-called \emph{point reduction algorithm} \cite{bayerremon} based on determining the tile to which a given point belongs in the hyperbolic upper half-plane. This is a magnificent improvement since, as far as the authors are aware, there are no known optimal decoders for general nonuniform constellations with sublinear complexity.
\item[$\bullet$]
We also discuss the optimization of the Fuchsian codes and propose a new design criterion, hence motivating further study on Fuchsian codes.
 \item[$\bullet$] Finally, we present an alternative method for constructing Fuchsian codes by certain parametrization of the integer tuples defining the M\"obius transformations used for the code construction.
\end{itemize}

Our interest in Fuchsian groups as a basis for code construction stems from a series of recent papers by Palazzo \emph{et al.} In \cite{brazilian_IEEE,brazilian_COAM,brazilian_franklin2,brazilian_franklin}, among others, various interesting connections between Fuchsian groups and signal constellation design are presented.  In \cite{brazilian_IEEE}, the authors construct Fuchsian groups suitable for signal constellation construction.  In  \cite{brazilian_franklin2}, the authors consider the unit disk model  of the hyperbolic half-plane as the signal space, and the noise is modeled as a hyperbolic Gaussian random variable. With the study of the hyperbolic geometry they construct a hyperbolic equivalent to  QAM and PSK constellations and point out that, when the channel model is hyperbolic\footnote{This is the case \emph{e.g.} in power transmission line communications \cite{59paper}.}, the proposed hyperbolic constellations provide higher coding gains than the classical euclidean variants. Building on this work, in \cite{brazilian_franklin} the authors construct dense tessellations and counting Dirichlet domains in tessellations of certain type.  In \cite{brazilian_COAM} the authors use units of quaternion orders to construct space-time matrices with the potential use case being wireless multi-antenna (MIMO) communications. We refer the reader to \cite{SRS,maxorder} as the early references to the use of division algebras and maximal orders in MIMO, and to \cite{viterbo2} for a more general introduction to the topic.

Although codes related to Fuchsian groups have been considered before, our construction is original in that it describes the complete construction and decoding process, whereas earlier work has largely  concentrated on the constellation design while giving little attention to the decoding and performance aspects.
Another key difference to the aforementioned works is  that we are studying codes on the \emph{complex plane} arising from quaternion algebras and Fuchsian groups, and  our aim is to apply the codes to the classical  (euclidean) channel models such as the aforementioned AWGN channel, with possible future extension to fading  channels \cite{viterbo1,viterbo2}. We do not use hyperbolic metric as our design metric, but use the Fuchsian group as a starting point to the  code generation. Nevertheless, our decoder will rely on hyperbolic geometry as opposed to the classical decoders based on euclidean geometry.

The paper is organized as follows. In what remains of this section, we will give some insight to AWGN channel decoding. In section \ref{algebraic} we provide the essential algebraic preliminaries. The Fuchsian code construction process as well as decoding via point reduction algorithm are introduced in Section \ref{construction}. Section \ref{complexitysec} provides a thorough decoding complexity analysis, showing that the decoding algorithm has logarithmic complexity.
We discuss the optimization of the proposed Fuchsian codes in Section \ref{optimization} as a motivation for further research. Conclusions and directions for further research are given in Section \ref{conclusions}. Finally, we present as an appendix an alternative method for constructing Fuchsian codes. This method is called for when the generators of the Fuchsian group are not known.

\subsection{Decoding in AWGN channels}
Let us discuss the decoding process in  AWGN channels before going to the actual code construction in more detail.
This decoding process, i.e., deciding on which codeword $x\in\Cc$ was transmitted given the received signal $y\in\C$ can be done in many different ways. An optimal decoding method is given by the \emph{maximum-likelihood} (ML) decoding, which decides on the codeword $\hat{x}$ having the smallest  squared euclidean distance to $y$,
\begin{equation}\label{ML}
\hat{x}=\textrm{arg\,min}|y-x|^2.
\end{equation}
This amounts to exhaustively enumerating the metric \eqref{ML} for all $x\in\Cc$, and comparing the values obtained in order to find the minimum.  The metric evaluations require $4|\Cc|$ arithmetic operations\footnote{By arithmetic operation we refer to addition, subtraction, multiplication, and division. These can all be considered constant time when we are computing with numbers having fixed precision. In \eqref{ML}, we need to compute the difference $y-x$, square the real and imaginary parts of the result and finally add them, $(\Re(y-x))^2+(\Im(y-x))^2$, which requires two multiplications, one subtraction and one addition per codeword.}, and to compare, we have to compute $|\mathcal{C}|-1$ differences. In total, this amounts to $5|\mathcal{C}|-1$ arithmetic operations. As far as the authors are aware, there are no other known optimal decoding methods for general nonuniform codes.

In \cite{WCC}, we have compared the error performance\footnote{The performance is typically measured as the relative frequency of decoding errors as a function of the signal-to-noise ratio (SNR). SNR is the ratio of the signal and noise powers, and is commonly used to  measure the channel quality.} of some Fuchsian codes to that of quadrature amplitude modulation (QAM) in order to get some preliminary insight as to how close  to these classical constellations we are able to get. We define an odd, symmetric square QAM constellation as
$$
2^{2r}\textrm{-QAM} = \{\pm a\pm b i\ |\ 1\leq a, b \leq 2^r-1,\, 2\not|\, ab\}\subset\Z[i].
$$
 This is a subset of the two-dimensional Gaussian integer lattice\footnote{By a lattice here we refer to a discrete abelian subgroup of $\C$. We refer to \cite{viterbo1} for a general introduction to lattice codes.} $\Z[i]$, hence its ML complexity can be written  as $5|\mathcal{C}|-1=5|S|^2-1$, where $S\subset\Z$ is the corresponding real pulse amplitude modulation (PAM) constellation,
 $$
 2^r\textrm{-PAM}=\{-(2^r-1),\ldots,-3,-1,1,3,\ldots,2^r-1\}\subset\Z.
 $$
 More generally, if we denote by $S$ the underlying real signaling alphabet $\subset\Z$ of a lattice code, the ML complexity $5|\mathcal{C}|-1=5|S|^\kappa-1$ grows exponentially with the lattice dimension $\kappa$.

For lattice codes, the ML complexity can be reduced by using \emph{lattice decoding}, which performs a closest lattice point search within a limited sphere centered at the received point $y$, while ignoring the fact that the codebook is a finite subset of the infinite lattice. The complexity of lattice decoding is hence independent of $|\Cc|$, and it actually turns out to be polynomial (cf. \cite{ViBu}) in $|S|$ for a given lattice and sphere radius. Unfortunately it also performs poorly compared to ML decoding. The performance can be improved by taking into account the code boundaries, often referred to as \emph{sphere decoding}, but this  again increases the complexity. Naturally, the worst case complexity of a sphere decoder is always upper bounded by the complexity of exhaustive search.

The complexity comparison between the QAM constellations and Fuchsian constellations is not straightforward since, in practice, one does not use ML, lattice or sphere decoder for decoding QAM in the single-input single-output (SISO) case (cf. Eq. \eqref{AWGN}). The difficulty of complexity comparison stems from the fact that, while the decoding complexity of the proposed Fuchsian codes largely arises from arithmetic operations, the decoding complexity of QAM in the SISO case is, in practice\footnote{We gratefully acknowledge Peter Moss (BBC Research \& Development) for sharing his knowledge and insights regarding AWGN channel decoding and complexity.}, a combination of arithmetic operations and memory usage due to maintenance of a look-up table. So for QAM, this finally boils down   to resource usage in a particular chip, the trade-off being memory vs. arithmetic operations. In addition, the estimate quality of the received signal is a parameter, since the amount of memory depends on the bit-resolution of the look-up table. In the literature, a look-up table is normally hand-waved as having negligible complexity, whereas in reality a very large table could still be highly inconvenient. Due to this comparison mismatch, we compare the complexity of Fuchsian codes to the ML decoding complexity $5|\Cc|-1$. This is also a more righteous comparison in the sense that, as noted before, nonuniform codes are not previously known to admit sublinear decoding complexity.  Indeed, one of the main contributions of this paper is  that our codes enable the use of a decoding algorithm with complexity that is logarithmic in the code size $|\Cc|$.

\begin{remark} We have chosen to use the number of arithmetic operations as the complexity measure. Another option would be to only count multiplications and divisions, since these are more complex than addition and subtraction. Nevertheless, both options yield very similar results. In addition, when the numbers involved in the arithmetic operations have known and predetermined precision, all arithmetic operations can be thought of as constant-time operations.
\end{remark}

\section{Algebraic preliminaries}
\label{algebraic}

In this section, we survey some facts on the arithmetic of quaternion algebras in order to construct a discrete group $\Gamma\in\SL(2,\R)$ and its fundamental domain in the complex upper half-plane. We mainly follow \cite{alsinabayer} and refer the reader to the well-known references \cite{katok} and \cite{vigneras} for more details.

\subsection{Quaternion algebras and Fuchsian groups}

For square-free $a,b\in\mathbb{Q}^*=\Q\setminus\{0\}$,  let $H=\left(\frac{a,b}{\mathbb{Q}}\right)$ be the quaternion $\mathbb{Q}$-algebra generated by $I$ and $J$
with the standard relations $I^2=a,J^2=b, K=IJ=-JI$. Up to isomorphism, we can assume $a$, $b$ are square-free nonzero integers. For $\omega= x+yI+zJ+tK\in H$, the conjugate is $\overline{\omega}=x-yI-zJ-tK$, and the reduced trace and the reduced norm are defined as
$$
\mathrm{Tr}(\omega)=\omega+\overline{\omega}=2x, \quad
\mathrm{N}(\omega)=\omega\overline{\omega}=x^2-ay^2-bz^2+abt.
$$
Let us denote by $\phi$ the following monomorphism of $\mathbb{Q}$-algebras:
\begin{equation}\label{mono}
\begin{array}{ccc}
\phi: \left(\dfrac{a,\, b}{\mathbb{Q}}\right) & \to & \M(2,\mathbb{Q}(\sqrt{a}))\\
x+yI+zJ+tK & \mapsto &
\left(\begin{array}{ccc} x+y\sqrt{a} &\phantom{x} & z+t\sqrt{a}\\
b(z-t\sqrt{a})&\phantom{x} & x-y\sqrt{a}\end{array}\right).
\end{array}
\end{equation}
Notice that for any $\omega\in H$, $\mathrm{N}(\omega)=\mathrm{det}\left(\phi(\omega)\right)$ and $\mathrm{Tr}(\omega)=\mathrm{Tr}\left(\phi(\omega)\right)$.

A quaternion $\mathbb{Q}$-algebra is either an algebra isomorphic to the matrix algebra $\M(2,\mathbb{Q})$ or a skew field, in the latter case typically called a division algebra. For any absolute value $|\phantom{x}|_p$ of $\mathbb{Q}$ attached to a place $p$, a place being either a prime number  or infinity,  $H_p:=H\otimes_{\mathbb{Q}}\mathbb{Q}_p$ is a quaternion $\mathbb{Q}_p$-algebra. For a local field $\mathbb{Q}_p$ or $\mathbb{R}$ there exists a unique quaternion division algebra. In the case of $\mathbb{R}$ it is the algebra of  \emph{Hamiltonian quaternions}.
If $H_p$ is a division algebra,  $H$ is called ramified at $p$.
The discriminant $D_H$ is defined as the product of the primes at which $H$ ramifies. Any quaternion algebra is ramified at a finite even number of places. Moreover, two quaternion $\mathbb{Q}$-algebras are isomorphic if and only if they have the same discriminant.

\begin{definition}
A rational quaternion algebra $H$ is called definite if it is ramified at $p=\infty$, and indefinite otherwise.
An indefinite quaternion algebra is called small ramified if $D_H$ is equal to a product of two distinct primes.
\end{definition}

An element $\alpha\in H$ is called integral if $\mathrm{N}(\alpha),\mathrm{Tr}(\alpha)\in \mathbb{Z}$. In general the set of integral elements in a quaternion algebra is not a ring.

A $\mathbb{Z}$-lattice  of $H$ is a finitely generated torsion-free $\mathbb{Z}$-module contained in $H$.
An order $\mathcal{O}$ of $H$ is a $\mathbb{Z}$-lattice and a ring such that $\mathbb{Q}\otimes \mathcal{O}\simeq H$.
Each order of a quaternion algebra is contained in a maximal order.
In an indefinite rational quaternion algebra, all the maximal orders are conjugate to each other (cf.\,\cite{vigneras}).

\begin{definition}
 Fix a quaternion algebra $H=\left(\frac{a,b}{\mathbb{Q}}\right)$ having discriminant $D> 1$, $D$ a product of an even number of primes, and a maximal order $\mathcal{O}\subset H$. Since $H$ is indefinite we can always assume $a>0$. Let us denote by $\Gamma(D,1)$ the image under the monomorphism $\phi$ (cf. Eq.(3)) of the group of units of reduced norm $1$ in $\mathcal{O}$, that is:
 $$\Gamma(D,1)=\phi( \{\omega\in\mathcal{O} \,\mid \, \mathrm{N}(\omega)=1\})\subseteq \M(2,\mathbb{Q}(\sqrt{a})).$$
\end{definition}

\begin{remark} The group $\Gamma(D,1)$ is a \emph{Fuchsian group}, a discrete subgroup of $\SL(2,\R)$.
Its elements will be called quaternion transformations. More details about its expression can be found in \cite{alsinabayer}.
\end{remark}

As a reference, consider the family of quaternion algebras $H=\left(\frac{p,-1}{\mathbb{Q}}\right)$. For any prime $p\equiv 3\mod 4$, it is an indefinite quaternion algebra of discriminant $2\cdot p$, and $\mathbb{Z}[1,I,J,(1+I+J+IJ)/2]$ is a maximal order. The group of quaternion transformations $\Gamma(2p,1)$ is equal to
$$
\left\{\gamma=\dfrac{1}{2}\left(\begin{array}{rr}\alpha & \beta\\-\beta' & \alpha'\end{array}\right) \,\mid \,  \alpha,\beta\in\Z[\sqrt{p}],\mathrm{det}(\gamma)=1,\alpha\equiv\beta\equiv\alpha\sqrt{p}\pmod{2}\right\},
$$
where $\alpha\mapsto \alpha'$ is the quadratic conjugation: $\alpha=a+b\sqrt{p}\in\Z[\sqrt{p}]$, $\alpha'=a-b\sqrt{p}$.

\begin{remark} The above construction is also valid for $D=1$. In this case, the corresponding group is the modular group $\SL(2,\Z)$.
\end{remark}

\subsection{Fundamental domains for quaternion groups}

Consider the complex upper half-plane $\mathcal{H}=\{z\in\mathbb{C}  \,\mid \,  \Im(z)>0\}$ endowed with the structure given by the hyperbolic metric (cf. \cite{alsinabayer}, \cite{katok}).

The group $\mathrm{SL}(2, \mathbb{R})$ acts on the complex upper half-plane $\mathcal{H}$ by M\"{o}bius transformations
and its action factorizes through $\mathrm{SL}(2, \mathbb{R})/\pm\Id$. Namely,
\begin{equation}\label{action}
\begin{array}{ll}
\text{for all } z\in \mathcal{H}, \quad & \gamma=\left( \begin{array}{cc} a_{11} & a_{12} \\ a_{21} & a_{22} \end{array}\right)\in \mathrm{SL}(2,\mathbb{R}),
\\
\gamma(z)=\dfrac{a_{11} z + a_{12}}{a_{21} z + a_{22}},\quad & \gamma(\infty)=\dfrac{a_{11}}{a_{21}}=\displaystyle\lim_{z\rightarrow \infty} \gamma(z).
\end{array}
\end{equation}

The Fuchsian groups are discrete subgroups of $\mathrm{SL}(2, \mathbb{R})$ and they have a proper and discontinuous action on $\mathcal{H}$.

\begin{definition} Let $\Gamma$ be a Fuchsian group. A connected closed hyperbolic polygon $\mathcal{F}$ in $\mathcal{H}$ is a fundamental domain for the action of $\Gamma$ on $\mathcal{H}$ if
\begin{itemize}
\item[a)] for any $z,z'$ in the interior of $\mathcal{F}$, if there exists $\gamma\in\Gamma$ such that $\gamma(z)=z'$, then $z=z'$ and $\gamma=\Id$,
\item[b)] for any $z\in\mathcal{H}$, there exists $z'\in\mathcal{F}$ and $\gamma\in\Gamma$ such that $\gamma(z)=z'$.
\end{itemize}
\end{definition}

By using fundamental domains with a pairing of the edges, a presentation of a Fuchsian group can be found.  Explicit fundamental domains for several Fuchsian groups of quaternion transformations $\Gamma(D,1)$ and their presentations can be found in \cite{alsinabayer}. Next, we include some examples of the presentations for the groups $\Gamma(6,1)$, $\Gamma(10,1)$ and $\Gamma(15,1)$ (cf. \cite{alsinabayer} Thm. 5.46, Thm. 5.47, Thm. 5.49),   as they will be used to exemplify the results of this paper. An algorithm applicable to a more general setting was stated in \cite{voi09}.

Each election of a fundamental domain for the action of a Fuchsian group $\Gamma(D,1)$ leads to a regular tessellation of the upper half-plane by hyperbolic polygons, which will be useful for the construction of Fuchsian codes.

\begin{example}\label{dom-pre61}
Consider the Fuchsian group $\Gamma(6,1)$, which will be used as the main example throughout the paper. A fundamental domain is displayed in Fig. 1 
and the corresponding presentation is the following:
$$
\Gamma(6,1)/\pm\Id=\langle g_1,g_2,g_3 \,\mid \, g_1^3=g_2^3=g_3^2=(g_1^{-1}g_3g_2)^2=1\rangle, \text{ where }
$$
$g_1:= \frac{1}{2}\left(\begin{array}{rr} 1+\sqrt{3} & 3-\sqrt{3} \\  - 3-\sqrt{3} & 1-\sqrt{3}\end{array}\right)$,
$g_2:= \frac{1}{2}\left(\begin{array}{rr}
1+\sqrt{3} & -3+\sqrt{3} \\
3+\sqrt{3} & 1-\sqrt{3}
\end{array}\right)$,
$g_3:= \left(
\begin{array}{rr} 0 & 1 \\   -1 & 0 \end{array}
\right)$.
\end{example}

\begin{figure}[H]\label{f-fundamentaldomain61}
     \includegraphics[width=0.55\textwidth]{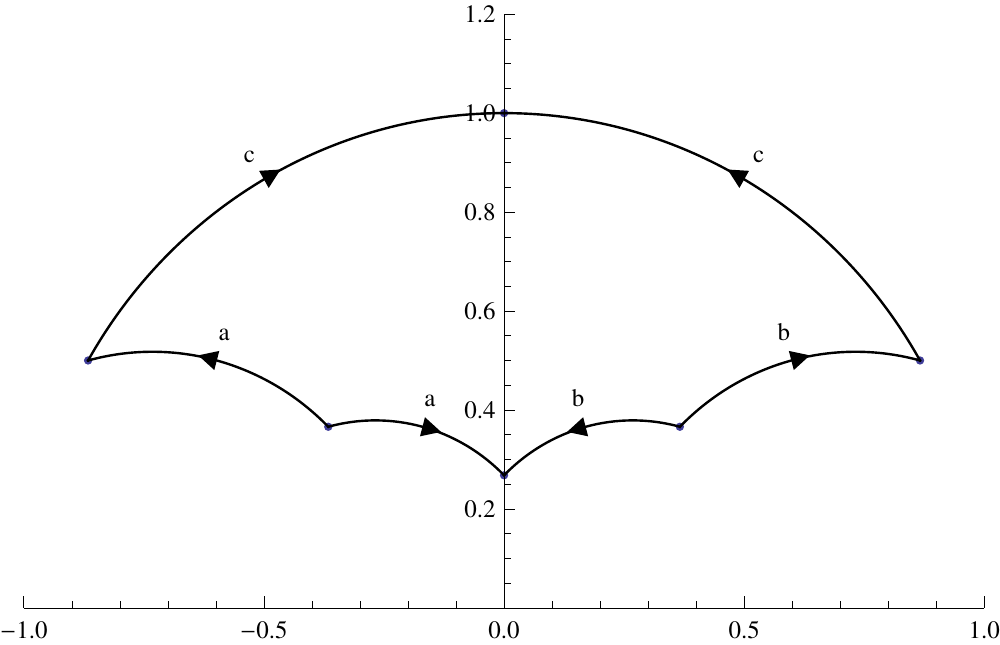}
     \caption[]{A fundamental domain for $\Gamma(6,1)$ with the pairing of the edges.}
\end{figure}

\begin{example}\label{dom-pre101}
A presentation for the Fuchsian group $\Gamma(10,1)$ is the following:
$$
\Gamma(10,1)/\pm\Id=\langle g_1,g_2,g_3 \,\mid \, g_1^3=g_2^3=(g_3^{-1}g_1)^3=(g_3^{-1}g_2)^3=1\rangle, \text{ where }
$$
$$g_1:= \frac{1}{2}\left(\begin{array}{cr} 1+\sqrt{2} & -1+\sqrt{2} \\  -5(1+\sqrt{2}) & 1-\sqrt{2}\end{array}\right),
\quad
g_2:= \frac{1}{2}\left(\begin{array}{cr} 1+\sqrt{2} & 1-\sqrt{2} \\  5(1+\sqrt{2}) & 1-\sqrt{2}\end{array}\right)
$$
$$
\text{ and } \quad
g_3:= \left(\begin{array}{cr} 3+2\sqrt{2} & 0 \\   0 & 3-2\sqrt{2} \end{array}\right).
$$
\end{example}

\begin{example}\label{dom-pre151}
A presentation for the Fuchsian group $\Gamma(15,1)$ is the following:
$$
\Gamma(15,1)/\pm\Id=\langle g_1,g_2,g_3 \,\mid \, (g_1g_3)^3=(g_3g_2^{-1}g_1g_2)^3=1\rangle, \text{ where }
$$
$g_1:= \frac{1}{2}\left(\begin{array}{cc} -4+3\sqrt{3} & -\sqrt{3} \\  5\sqrt{3} & -4-3\sqrt{3}\end{array}\right)$,
$g_2:= \frac{1}{2}\left(\begin{array}{cc} 3 & 1 \\   5 & 3 \end{array}\right)$,
$g_3:= \left(\begin{array}{cc} 2+\sqrt{3} & 0 \\   0 & 2-\sqrt{3} \end{array}\right)$.
\end{example}

The construction of fundamental domains is based on the use of isometric circles, a geometric object that will be used in the implementation of our decoding algorithm.

\begin{definition}
Given $\gamma=\begin{pmatrix}a_{11} & a_{12} \\ a_{21} & a_{22} \end{pmatrix}\in\Gamma$ such that $a_{21}\not=0$, the \emph{isometric circle} of $\gamma$ is
$$
I(\gamma)= \{z\in\mathcal{H} \,\mid \,  |a_{21} z+a_{22}| =1\}.
$$

The center and the radius of $I(\gamma)$ are the real numbers $-a_{22}/a_{11}$ and $|1/a_{21}|$, respectively.
\end{definition}

\begin{definition}\label{Mdef}
For a Fuchsian group $\Gamma$ and a fixed fundamental domain $\mathcal{F}(\Gamma)$ as above, let us denote by $G$ the set of elements in $\Gamma$ such that the edges of $\mathcal{F}(\Gamma)$ are included in the set of isometric circles defined by the elements of $G$. Let us denote $M=|G|$.
\end{definition}

\begin{remark}\label{rem-G}
We will split $G$ in two sets denoted by $G^{\text{int}}$ and $G^{\text{ext}}$ in such a way that the fundamental domain $\mathcal{F}(\Gamma)$ is the closure of
 $$ \bigcap_{\gamma\in G^{\text{ext}}} \ext(I(\gamma)) \bigcap_{\gamma\in G^{\text{int}}} \inte(I(\gamma)), $$
 where $\ext(I(\gamma))$ and $\inte(I(\gamma))$  denote the exterior and the interior of the isometric circle $I(\gamma)$, respectively.  The presentation of the group arises from the pairing of the edges; thus we can assume the generators of $\Gamma$ are included in $G$.
\end{remark}
This is illustrated in Fig. 2, where we have depicted a fundamental domain for $\Gamma(6,1)$ (cf. Ex. 1). The isometric circles corresponding to the edges of the hyperbolic polygon are displayed, labeled in terms of the generators of the group. In this example,
$$ G^{\text{ext}}=\{g_1,g_1^{-1},g_2,g_2^{-1}\}, \quad G^{\text{int}}=\{ g_3 \}, \qquad \text{and } M=5.$$

\begin{figure}[H]\label{f-fundamentaldomain61}
     \includegraphics[width=0.55\textwidth]{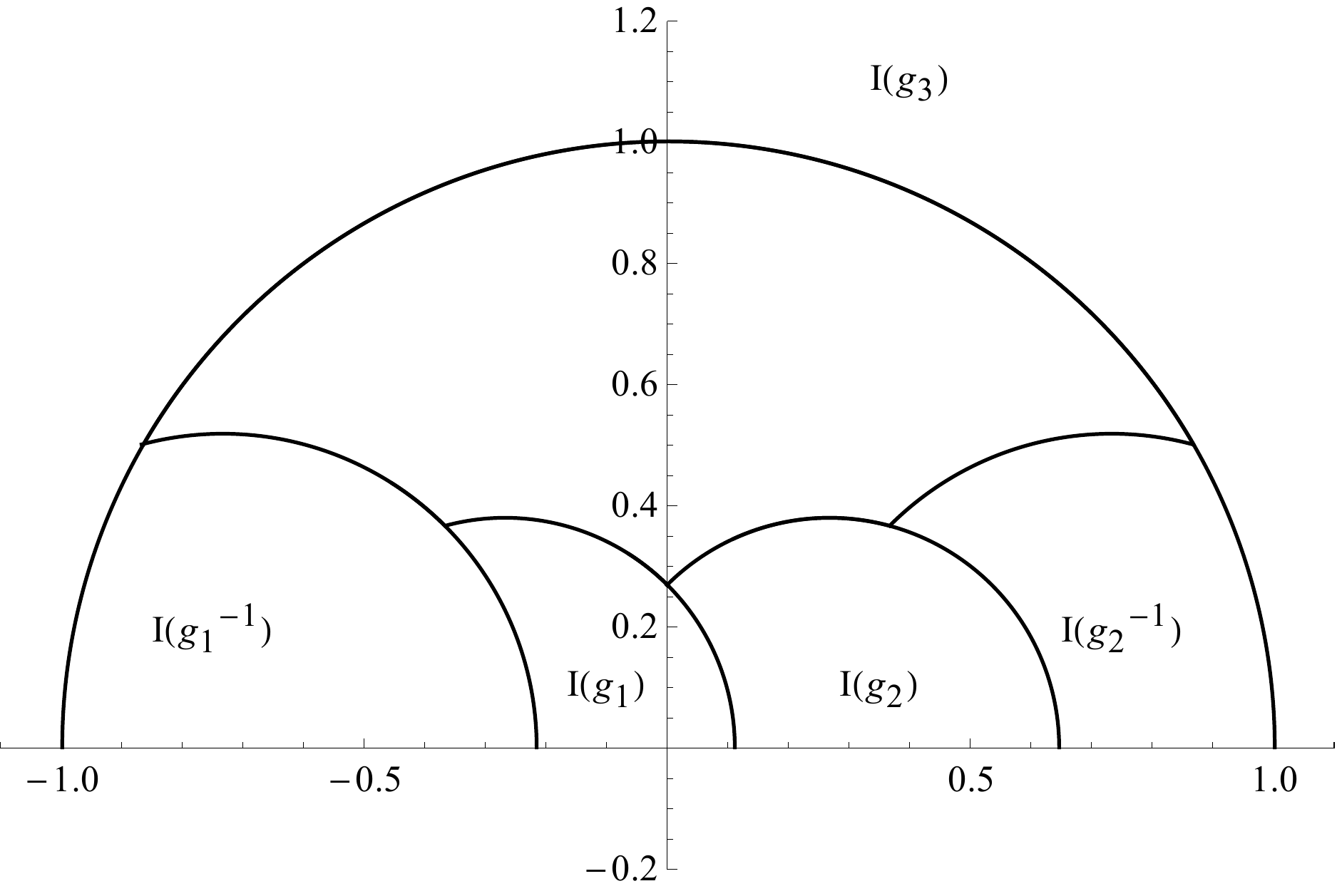}
     \caption[]{A fundamental domain for $\Gamma(6,1)$  labeled with isometric circles.}
\end{figure}

\section{Construction of Fuchsian codes}
\label{construction}

In this section, we will show in detail how to construct and decode Fuchsian codes in an AWGN channel. The first subsection describes our proposal for the construction of a new family of Fuchsian codes. The second subsection introduces the point reduction algorithm (PRA) \cite{bayerremon}, which will be used for decoding in the third subsection.

In what follows $\Gamma$ will be a Fuchsian group $\Gamma(D,1)$ with $D>1$ a product of an even number of primes. In fact, our construction could be formulated more generally in terms of compactness for  any \emph{cocompact} Fuchsian group $\Gamma$.

\subsection{Construction}

Let us now fix a Fuchsian group $\Gamma=\Gamma(D,1)$, a fundamental domain $\mathcal{F} = \mathcal{F}(\Gamma)$ and an ordered set of generators $G$.  We choose  a point $\tau$ in the interior
of $\mathcal{F}$; this condition ensures that $\gamma(\tau)\not=\tau$ for all $\gamma\in \Gamma\setminus \{\pm\Id\}$.

The first step in the code construction is to choose $N$ elements in $\Gamma$. We denote this finite set by $S_{\Gamma}=\{ \gamma_1, \ldots \gamma_N \}$. The first elements can be directly taken to be the generators of the group, and the rest will be expressed as products of generators.
Later on in this section we will discuss the choice of the elements $\gamma \in S_\Gamma$ in more detail.

Considering the action of the group $\Gamma$ in the complex upper half-plane $\mathcal{H}$ defined in \eqref{action}, we obtain the points $\gamma_1(\tau), \ldots \gamma_N(\tau)$ in $\mathcal{H}$. These will serve as the first points to be included in our codebook. We can double the number of points by expanding to the lower half-plane in a natural way by including the opposites $-\gamma(\tau)$. This has two advantages:
\begin{enumerate}
\item Duplicating the code size in this way does not increase the average/maximum energy (cf. Eq. \eqref{energy}) of the constellation, since $|\gamma(\tau)|=|-\gamma(\tau)|$.
\item The complexity of our decoding algorithm (Sec. \ref{dec}) is related to the maximum number of generators $g_i$ in the presentation of $\gamma$ as a product of generators. Hence, it is favorable to construct the code by using as few different matrices $\gamma$ as possible to avoid having to involve more generators than necessary.
\end{enumerate}
Table 1
below summarizes the construction process.

\begin{table}[H]\label{construction-figure}
\caption{Sketch of the code construction process.}
$
\boxed{
\begin{array}{c}
\boxed{
\begin{array}{c}
\text{FIX:} \\
\text{the Fuchsian group } \Gamma \text{ and a} \\
\text{fundamental domain  } \mathcal{F}(\Gamma)
\end{array}
      }
\\
\downarrow
\\
\boxed{
\begin{array}{c}
\text{CHOOSE:} \\
S_{\Gamma}=\{\gamma_1, \ldots \gamma_N \}\subset \Gamma \\
\tau \text{ in the interior of } \mathcal{F}(\Gamma)
\end{array}
}
\\
\downarrow
\\
\boxed{
\begin{array}{c}
\text{CONSTRUCT the CODE:} \\
\mathcal{C}=\{\pm\gamma_1(\tau), \ldots \pm\gamma_N(\tau)\} \\
|\mathcal{C}|=2N
\end{array}
}
\end{array}
}
$
\end{table}

Formally, we define a Fuchsian code as follows.
\begin{definition}
Let  $\Gamma$ be a Fuchsian group defined as above. Given a fundamental domain $\mathcal{F}(\Gamma)$, a set $S_{\Gamma}$,  and a point $\tau$ in the interior of $\mathcal{F}(\Gamma)$, we define the associated \emph{Fuchsian code} as $\mathcal{C} = \left\{\pm\gamma(\tau) \, \mid \,  \gamma \in S_\Gamma\right\}\subseteq \mathbb{C}$.

The point $\tau$ is called the \emph{center} of the code. For a fixed code size $q=|\mathcal{C}|=2N$, the corresponding constellations will be referred to as \emph{nonuniform Fuchsian constellations},  $q$-NUF in short. \end{definition}

\begin{remark} Nonuniform constellations have been used already in early-state signal transmission, \emph{e.g.}, in the so-called  \emph{codec} transmission, and are present more recently in the DVB-NGH standard. Currently, the use of certain nonuniform constellations is being discussed and seriously considered for the multi-antenna extension of the terrestrial DVB standard (DVB-T2). While in this paper we are considering the AWGN channel, which in the context of DVB is mainly relevant for satellite transmission and the theoretical understanding of the codes, our aim is to generalize this framework to fading channels and to design codes directly applicable to the general DVB framework. For more information, see \cite{dvb}.
\end{remark}
\begin{example}
Let us consider the Fuchsian group $\Gamma(6,1)$ and the fundamental domain displayed in Table 2 below, which collects the data used for generating Fuchsian codes of size 4, 8, and 16, in terms of the generators (cf. the presentation given in Example \ref{dom-pre61}. The explicit values of the codewords in $\mathcal{C}$ are also included. The 16-NUF constellation is displayed in Fig. 3. For brevity, the codes arising from Example \ref{dom-pre101} and Example \ref{dom-pre151} are given in terms of the center and generators.
\end{example}
\def\arraystretch{1.5}
\begin{table}[H]\label{table-data61-4816}
\caption{Explicit choices for $\tau$ and $S_{\Gamma}$ and the list of resulting codewords in $\mathcal{C}\cap\mathcal{H}$, $q=|\mathcal{C}|=4,8,16$, for  $\Gamma(D,1)$.}
\begin{tabular}{|c|c|}
\hline
$\mathbf{\Gamma(6,1)}$ & $\tau= \dfrac12 i$ \\
\hline
  $q=4$ & $S_{\Gamma}=\{\textrm{Id}, \, g_1^{-1}\}$  \\
\hline
 Codewords & $\frac{i}{2} , -\frac{5}{7} (-3+2 \sqrt{3}) -  \frac{4}{7} i (-2+\sqrt{3}) $\\
\hline
\hline
$q=8$ &   $S_{\Gamma}=\{\textrm{Id}, \, g_1^{-1}, \, g_2^{-1}, \, g_3\}$ \\
\hline
 Codewords
 & $\frac{i}{2} , -\frac{5}{7} (-3+2 \sqrt{3}) -  \frac{4}{7} i (-2+\sqrt{3}) $\\
 & $\frac{5}{7} (-3+2 \sqrt{3}) -  \frac{4}{7} i (-2+\sqrt{3})  , 2i$\\
\hline
\hline
$q=16$ &  $S_{\Gamma}=\{\textrm{Id}, \, g_1^{-1}, \, g_2^{-1}, \, g_3, \, g_1, \, g_2, g_1^{-1}g_3, g_2g_3\}$ \\
\hline
 Codewords  & $\frac{i}{2} , \,  -\frac{5}{7} (-3+2 \sqrt{3}) -  \frac{4}{7} i (-2+\sqrt{3}) $\\
                             & $\frac{5}{7} (-3+2 \sqrt{3}) -  \frac{4}{7} i (-2+\sqrt{3})  ,\,  2i$\\
                              & $\frac{1}{193} (96-131 \sqrt{3}) + \frac{4}{193} i \left(14+\sqrt{3}\right), \,  -\frac{1}{193} (96-131 \sqrt{3}) + \frac{4}{193} i \left(14+\sqrt{3}\right) $\\
                               & $-\frac{5}{13} (-3+2 \sqrt{3}) -\frac{4}{13} i (-2+\sqrt{3}), \frac{5}{13} (-3+2 \sqrt{3}) -\frac{4}{13} i (-2+\sqrt{3})$\\
\hline

\hline
$\mathbf{\Gamma(10,1)}$ & $\tau= \frac{2}{5}i$  \\
\hline
$q=4$ &  $S_{\Gamma}=\{\Id, \, g_1^{-1}\}$  \\
\hline
$q=8$ &  $S_{\Gamma}=\{\Id, \, g_1^{-1}, g_2^{-1}, g_1\}$  \\
\hline
$q=16$ &  $S_{\Gamma}=\{\Id, \, g_1^{-1}, g_2^{-1}, g_1,g_2, g_1g_2^{-1}, g_2g_1^{-1}, g_3^{-1}\}$  \\
\hline
\hline
$\mathbf{\Gamma(15,1)}$ & $\tau= \frac{9}{10}i$  \\
\hline
$q=4$ &  $S_{\Gamma}=\{\Id, \, g_2\}$  \\
\hline
$q=8$ &  $S_{\Gamma}=\{\Id, \, g_2, g_1, g_2^{-1}\}$  \\
\hline
$q=16$ &  $S_{\Gamma}=\{\Id, \, g_2, g_1, g_2^{-1}, g_1^{-1}, g_3^{-1}, g_2^{-1}g_1g_2, g_2^{-1}g_1^{-1}g_2\}$  \\
\hline

\end{tabular}
\end{table}

\begin{figure}[h] \label{f-codewords61-16}
     \caption[]{Example of 16-NUF constellation for $\Gamma(6,1)$.}
     \scalebox{1}{
     \includegraphics[width=0.40\textwidth]{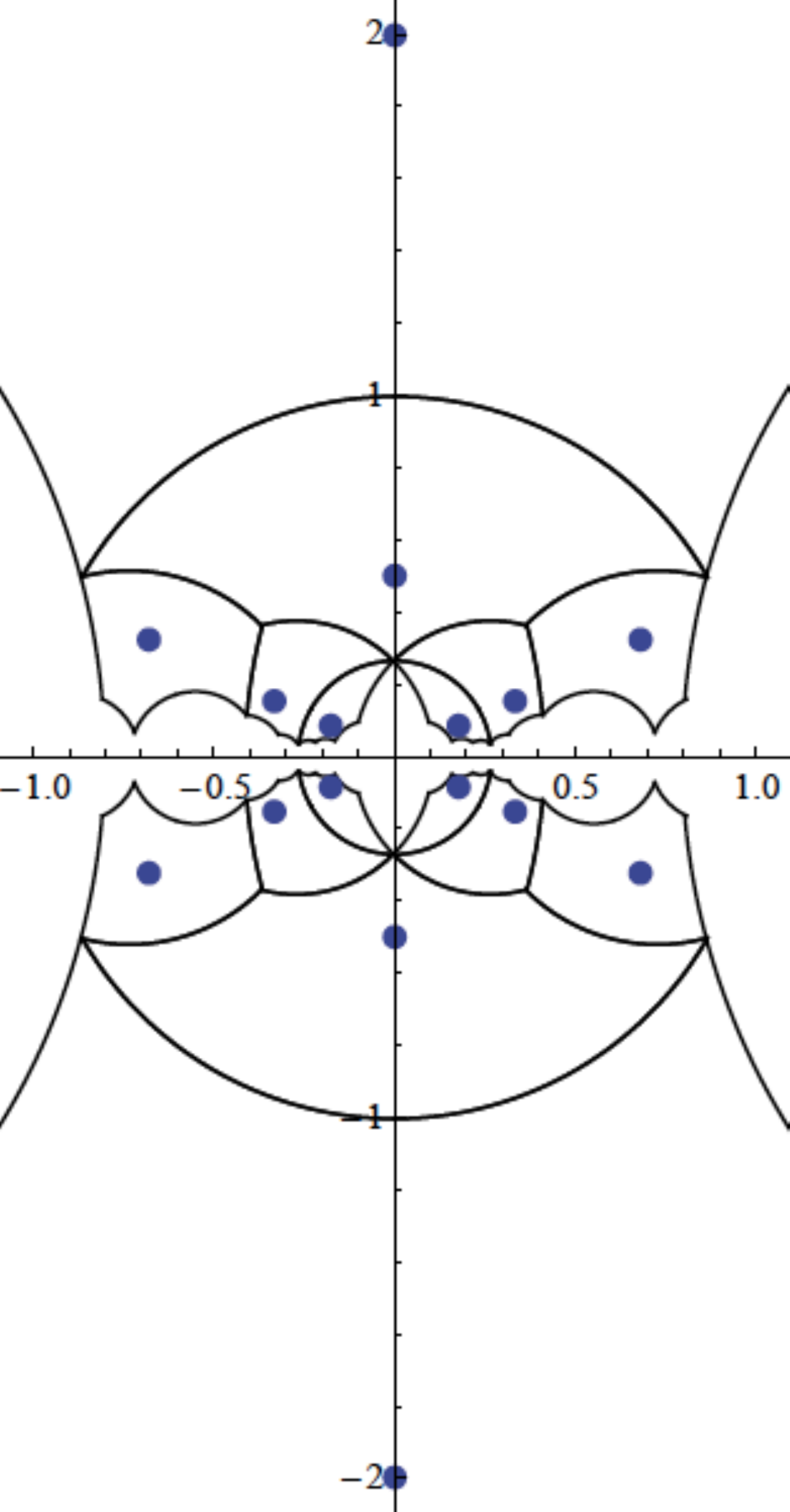}
                }
\end{figure}


\subsection{The point reduction algorithm (PRA)}

In order to decode the Fuchsian codes in AWGN channels, we first show that this problem is equivalent to certain point reduction in the upper half-plane. As we saw in the definition of a fundamental domain, any point in the complex upper half-plane has its equivalent in the fundamental domain. Using this fact and a finite number of M\"obius transformations we will be able to recover the transmitted points. To this end, we will employ the so-called point reduction algorithm. In what follows, we will explain the general guidelines of the algorithm, originally presented in \cite{bayerremon}.

Given a cocompact Fuchsian group $\Gamma$ and a fundamental domain $\mathcal{F}(\Gamma)\subseteq \mathcal{H}$,
the algorithm reduces a given point $z\in \mathcal{H}$ to a point $z_0\in \mathcal{F}$, and yields a transformation $t\in \Gamma$ such that $t(z)=z_0$. Shortly:

\begin{tabular}{ll}
\textbf{Input:}  & a point $z\in \mathcal{H}$.\\
\textbf{Output:} & a point $z_0\in \mathcal{F}$, and a matrix  $t\in \Gamma$ such that $t(z)=z_0$.
\end{tabular}


Next, consider the finite ordered set $G$ derived from the fundamental domain by taking into account the isometric circles (cf. Section \ref{algebraic}), $G=G^{\text{int}}\cup G^{\text{ext}}$. The following algorithm is adapted from the point reduction algorithm in \cite{bayerremon}, where a proof of correctness is given, based on results in \cite{katok}. The complexity of the point reduction algorithm is treated in full detail in Section \ref{complexitysec}.

\begin{algorithm}\label{algorithm}
\textbf{ALGORITHM}\\
  \begin{itemize}
\item[]  \textbf{Step 1} Initialize: $z_0 = z$ and $t = \mathrm{Id}$.
\item[]  \textbf{Step 2} Check if $z_0 \in \mathcal{F}$.
         \newline \hphantom{\textbf{Step 2\,}}
         If $z_0 \in \mathcal{F}$, return $z_0$ and $t$. Quit.
         \newline \hphantom{\textbf{Step 2\,}}
        If $z_0 \not\in \mathcal{F}$, return $g\in G$ such that:
        \newline \hphantom{\textbf{Step 2\,}}
        \hspace{0.5cm}$z_0 \in \mathrm{int}(\mathrm{I}(g))$, if $g\in G^{\text{ext}}$,
                \newline \hphantom{\textbf{Step 2\,}}
         \hspace{0.5cm}$z_0 \in \mathrm{ext}(\mathrm{I}(g))$ if $g\in G^{\text{int}}$.
\item[]  \textbf{Step 3}    Compute $z_0 = g (z_0)$ and $t = g\cdot t$. Go to Step 2.
  \end{itemize}
\end{algorithm}

The following definition of \emph{depth} will help us to choose codewords that contribute as little as possible to the complexity of the above algorithm.

\begin{definition}\label{depth}
Let $\mathcal{F}(\Gamma)$ be a fundamental domain of a Fuchsian group $\Gamma$ similarly as above.
We define the \emph{depth  of a point} $z \in \mathcal{H}$ as the number $\ell (z)$  of iterations of the algorithm, namely iterations of step 3, to reduce $z$ to a point $z_0\in \mathcal{F}$.
The \emph{depth of a matrix} $\gamma \in \Gamma$ is defined as the number $\ell(\gamma)$  of iterations of the algorithm to reduce the point $\gamma(p)$ to the point $p$, for any point $p$ in the interior of $\mathcal{F}(\Gamma)$. It is straightforward to see that this number is independent of the choice of p, so it is well-defined. The \emph{depth of a set} $S\subseteq \Gamma$ is defined as $\ell(S) = \max \{\ell(\gamma) \, \mid \, {\gamma\in S}\}$.
\end{definition}

We immediately observe that the identity element $\mathrm{Id} \in \Gamma$ satisfies $\ell(\mathrm{Id}) = 0$.
We can control the depth as follows.

\begin{lemma}
   Let $\Gamma$ be a cocompact Fuchsian group, $S$  a finite ordered set of generators of $\Gamma$ such that for any $\gamma\in S$ also $\gamma^{-1}\in S$. Then $ \ell(S)=1$.
\end{lemma}

\begin{lemma}\label{lem-depth}
   If $\ell (\gamma)=\kappa$, then $\gamma$ can be written as the product of $\kappa$ elements in $G$.
\end{lemma}

\subsection{Decoding of Fuchsian codes}\label{dec}

Let $\tau$ be the center of the code, $x\in \Cc\subset\C$ the transmitted codeword  and $y$ the  received signal,  $y = x + w=\gamma(\tau) + w$, where $w$ is the Gaussian noise. In order to remain in the upper half-plane, we initialize the algorithm with $z_0=y$, if $\Im(y)>0$, and with $z_0=-y$, if $\Im(y)<0$. Since $\R$ has measure zero in $\C$, the case $\Im(y)=0$ occurs with probability zero.

Let us first consider the case $\Im(y)>0$.
We use the above point reduction algorithm to obtain a point $\gamma'(y) \in \mathcal{F}$, and store the matrix
$\gamma'$.
The decoded word will be $\gamma'^{-1}(\tau)$.
If the channel quality is sufficient, we shall have $\gamma' =\gamma^{-1}$ and we can recover $x$.

If $\Im(y)<0$, then we reduce the point $-y$; thus, we obtain a matrix $\gamma'$ such that
$$
(\gamma'\circ n)(y) \in \mathcal{F}, \quad \text{ where }
n=\left(\begin{array}{rc}
-1 & 0\\
0& 1\\
\end{array}\right), n(y)=-y.
$$
The decoded word will be $\gamma''^{-1}(\tau)$,  where $\gamma''= \gamma'\circ n$.
Again, with sufficient channel quality, we shall have $(\gamma'')^{-1} =n \circ \gamma$, and we can recover
$x$.

Next we will apply the properties of PRA to the choice of the codewords.

\begin{definition}\label{depth}
Let again $\mathcal{F}(\Gamma)$ be a fundamental domain of a Fuchsian group $\Gamma$ defined as above.
We define the \emph{depth of the code $\Cc$} as
$$
\ell(\Cc) = \mathrm{max}\{\ell(x) \,|\, x \in\Cc\}.
$$
\end{definition}

Now we have all the tools to construct a set $S_\Gamma$ with good properties.
We fix the point $p=\tau\in \mathcal{F}$ to be  the center of the code in order to consider the depth of elements in $\Gamma$.
We define
$$
S^\kappa_\Gamma = \{\gamma \in \Gamma\,|\, \ell(\gamma) \leq \kappa\}
\quad \text{ and } \quad  \theta_\kappa = | S^\kappa_\Gamma|.
$$

\begin{remark}
Fuchsian groups 
are infinite groups and $\theta_{\kappa - 1} < \theta_\kappa$.  The study of the values $\theta_\kappa$ is done by using the growth function. Results on cocompact Fuchsian groups were presented in \cite{Can84}.
\newline
Therefore, our task is to search for the smallest $\kappa$ such that $\theta_{\kappa - 1} < |\mathcal{C}| \leq \theta_\kappa$.
Then, for the code size $|\Cc|$, we will choose $S_\Gamma$ such that
$$
    S_\Gamma \subseteq S_\Gamma^{\kappa}.
$$
\end{remark}
%

%
%
The advantages of the above condition are twofold:
\begin{enumerate}
\item This choice will optimize the running time of the algorithm since our code will consist of matrices with minimal depth. 
\item We are considering good tiles for the codewords as they are obtained from the edges of the fundamental domain in a systematic way, hence keeping them as close as possible to the initial tile.
\end{enumerate}
This criterion was used to build the example of 16-NUF constellation. The first advantage is reflected in the data in Table 2, and the second one in the tiles in Fig.3.


\begin{remark} We have constructed the codes starting from a fundamental domain and the group generators. However, in the case when the explicit domain or generators cannot be computed, we need to think of other construction methods. This can be done by a general parametrization to come up with a desired number of matrices $\gamma_i$, and will be explained in detail in Appendix. We also point out that in this case, since the point reduction algorithm requires the information about the fundamental domain and the generators,   modifications to the algorithm are needed. Naturally, one can always choose to do ML decoding.
\end{remark}

\section{Complexity}
\label{complexitysec}

In this section we see how the properly discontinuous character of the action of a Fuchsian group $\Gamma$ implies fast decoding. Let 
$\mathcal{C} = \left\{\pm\gamma(\tau) \, \mid \,  \gamma \in S_\Gamma\right\}$ be the codebook
.
Since we have chosen $\tau$ in the interior of $\mathcal{F}$, all the points in the codebook are indeed distinct, so $|\mathcal{C}|=2N$.

Consider $G$ the set of elements in $\Gamma$ defined in Section \ref{algebraic} according to the election of the fundamental domain, $G=G^{\text{int}}\cup G^{\text{ext}}$.

\begin{proposition} The complexity of the decoding algorithm for a Fuchsian code $\mathcal{C}$,
 in number of arithmetic operations (i.e. sums, differences, products, and divisions)
is
$$
r_{\mathcal{C}}\leq \ell(\Cc)(5M+14) +5M+7,
$$
where $M$ is defined as in Def. \ref{Mdef}. Hence, $M$ is a constant\footnote{To give some idea as to how big the constant $M$ is, we have $M=5\,,6,\,8$ for $\Gamma(6,1)$, $\Gamma(10,1)$, $\Gamma(15,1)$, respectively.} independent of the code size $|\mathcal{C}|$.
\label{complexity}
\end{proposition}
\begin{proof} First we take in account Steps 1-3 for the PRA.

\textbf{Step 1.} The algorithm initializes $z_0$ to be either the channel output  $y$ or $-y$, depending on the sign of $\Im(y)$. The accumulator matrix $t$ is set to be the identity. These initializations do not imply arithmetic operations.

\textbf{Step 2.} This step consists of checking whether the point $z_0$ belongs to the fundamental domain.  Since the fundamental domain is given in terms of the intersection of the exteriors or interiors of the isometry circles (cf. Remark \ref{rem-G}), this requires to check recursively if the point belongs to the interior of $I(\gamma)$ for $\gamma\in G^{\text{int}}$, and to the exterior of $I(\gamma)$ for $\gamma\in G^{\text{ext}}$.
Hence, if the point belongs to the fundamental domain, this step will finish after checking the $M$ isometry circles corresponding to $G$, and the algorithm will stop. Otherwise, it will find $g\in G$ such that the condition on the isometry circle $I(g)$ is not satisfied. In the worst case, we are checking $M$ isometry circles.
To determine whether or not a given complex number belongs to an isometry circle implies  performing $5$ arithmetic operations ($2$ real multiplications, $1$ sums and $2$ differences). Hence, this step takes $5M$ arithmetic operations.

The matrix of $g$ is stored in this step. This does not imply arithmetic operations. In fact we can avoid storing matrices at this step, because $G$ is an ordered set and to store the index will be enough.

\textbf{Step 3.} In case the point does not belong to the fundamental domain, the algorithm continues in this third step. Here, once we have identified an element  $g$ such that the interior of its isometry circle contains the point (by the previous step), we multiply the accumulator $t$ by $g$, which requires $12$ arithmetic operations ($2$ products and $1$ sum per entry), and update $z_{k+1}=g(z_k)$, which accounts for $7$ arithmetic operations, $19$ arithmetic operations all told. Then we go to Step 2, but now we can avoid checking with the element $g$ just applied, which means at most $5(M-1)$ operations.

Thus, given a point, the PRA returns the element $\gamma'$ with $\ell(\mathcal{C})$ iterations, which means applying  Step 2 once and Step 3 $\ell(\Cc)$ times, followed by Step 2. In total at most $\ell(\Cc)(5M+14) +5M$ operations.

Finally, the decoded word is obtained by computing $\gamma'^{-1}(\tau)$, \emph{i.e.}, 7 arithmetic operations, since $\det(\gamma')=1$.

Summarizing, we have $r_{\mathcal{C}}\leq \ell(\Cc)(5M+14) +5M +7$.
\end{proof}

A study of Fuchsian codes for the group $\Gamma(6,1)$ was carried out in order to compare the growth of the depth, $\ell(\Cc)$, with the growth of the code size, $|\mathcal{C}|$, prior to developing these theoretical results. In the following table, the growth of $\ell(\Cc)$ and the growth of $|\mathcal{C}|$ are compared.

\begin{table}[H]
\label{comparison}
\caption{Experimental relationship between the depth $\ell(\mathcal{C})$ and the size $|\mathcal{C}|$,  for Fuchsian codes $\mathcal{C}$ attached to $\Gamma(6,1)$.}
\begin{tabular}{|c|c|c|c|c|c|c|c|c|c|} \hline
$   |\mathcal{C}|$   & 4  & 8  &  16  & 32 &  64  & 128 &  256 & 512  & 1024 \\ \hline
$\ell(\mathcal{C})$  & 1  & 1  &   2  &  3 &   3  &  4  &   5  &   5  &  6   \\ \hline
\end{tabular}
\end{table}


\begin{proposition}Let $\Gamma$ be a Fuchsian group containing a non-abelian free subgroup. Then
$$
\ell(\mathcal{C})\leq\kappa_0\left( \frac{\log(|C|+2)}{\log(2)} -2\right),
$$
where $\kappa_0\geq 1$ is a constant depending only on the Fuchsian group.
\end{proposition}

\begin{proof}
Let $h_1,h_2\in \Gamma$ such that $\langle h_1,h_2 \rangle\subseteq \Gamma$ is a non-abelian free subgroup and denote
$\kappa_0=\max \{\ell(h_1),\ell(h_2)\}.$
\newline
Consider
$S_t=\{ h_{i_1}h_{i_2}\cdots h_{i_m}\mid h_{i_j}\in\{h_1,h_2\}, \, m\leq t\} \subset\langle h_1,h_2 \rangle $.
\newline
We have $|S_t|=2^{t+1}-1$, because of the non-abelian free character.
\newline
Since it is clear that $\ell(S_t)=t\kappa_0$, we have
$$S_t\subset S^{t\kappa_0}_{\Gamma}, \qquad \text{ thus } \quad |S^{t\kappa_0}_{\Gamma}|\geq 2^{t+1}-1.$$

Taking in account the duplication process, a Fuchsian code $\mathcal{C}$ can be constructed  in such a way  that $|\mathcal{C}|\geq 2(2^{t+1}-1)$, and the depth $\ell(\mathcal{C})\leq t\kappa_0$.
It follows then that
$$
t\leq  \frac{\log(|C|+2)}{\log(2)} -2, \qquad \text{ then } \quad
\ell(\mathcal{C})\leq \kappa_0 \frac{\log(|C|+2)}{\log(2)} -2.
$$
\end{proof}

\begin{remark}
Notice that if we were able to use a fundamental domain of $\Gamma$ in such a way that $h_1,h_2\in G$, then $\kappa_0=1$. For the Fuchsian groups $\Gamma(D,1)$ the choice of an element $h_1$ can be done by using the principal homothety of $\Gamma$, studied in \cite{alsinabayer}, related to a fundamental unit of the real quadratic field $ \mathbb{Q}(\sqrt{a})$ (cf. Eq. (3)).
Estimation of $\kappa_0$ in general is a difficult problem; some partial results have been proved in, \emph{e.g.}, \cite{Tal10}.
\end{remark}

By using the above propositions we arrive at the following upper bound for the complexity.

\begin{cor} Let $\mathcal{C}$ be a Fuchsian code attached to a group $\Gamma$ containing a non-abelian free subgroup. Then the complexity can be upper bounded as
$$r_{\mathcal{C}}\leq \overline{r_{\mathcal{C}}} =  \kappa_0 (5M+14) \left(\frac{\log(|C|+2)}{\log(2)} -2\right) +5M +7,$$
where $\kappa_0$ is a constant depending only on the group $\Gamma$ and the choice of its fundamental domain.
\label{complexbound}
\end{cor}

Since the Fuchsian groups considered in this paper are non-elementary, they have a free a non-abelian subgroup (\cite{fine}). In particular, this free non-abelian subgroup has at least two generators, for otherwise it would be cyclic. Hence, the complexity bound in corollary \ref{complexbound} holds for our constructions. In fact, for $\Gamma(6,1)$ an experimental value of $\kappa_0=1$ is obtained, which leads to a very interesting bound for the complexity, especially when  $|\mathcal{C}|\geq2^5$.

Taking into account the experimental value of $\kappa_0=1$ for $\Gamma(6,1)$, we compare the decoding complexity by using the point reduction algorithm to the complexity of ML decoding, \emph{i.e.}, exhaustive comparison of the received signal with all the elements in the codebook, and choosing the closest one. As mentioned earlier, the ML method consists of $5|\mathcal{C}|-1$ comparisons. To do this, we use the bound in Corollary \ref{complexbound}, taking into account that $M=5$ for $\Gamma(6,1)$. In Table 4, we depict the complexity reduction for different code sizes $|\mathcal{C}|$. The entries of the table give the \emph{complexity reduction percentage} (CRP),
$$
CRP_{|\mathcal{C}|}=100\left(\frac{(5|\mathcal{C}|-1)-\overline{r_{\mathcal{C}}}}{5|\mathcal{C}|-1}\right)
$$
for $|\mathcal{C}|=4,8,16,64,256,512,$ and $1024$. Note that a zero entry means that it is favorable, in terms of complexity, to use  ML decoding instead of the PRA.

\begin{table}[H]
\label{comparison}
\caption{Complexity reduction percentage achieved with the PRA decoding of the Fuchsian codes attached to $\Gamma(6,1)$, compared to ML decoding for code sizes $|\mathcal{C}|=4$,$8$,$16$,$64$,$256$,$512$,$1024$.}
\begin{tabular}{|c|c|c|c|c|c|c|}
\hline  $CRP_4$ & $CRP_{8}$ & $CRP_{16}$ & $CRP_{64}$ & $CPR_{256}$ & $CRP_{512}$ & $CRP_{1024}$\\
\hline  $0$    & $0$       & $0$        & $5.79$ & $70.40$ & $83.68$ & $91.08$  \\
\hline
\end{tabular}
\end{table}

\begin{remark}
Since we have used the complexity upper bound $\overline{r_{\mathcal{C}}}$, the above CRPs are somewhat pessimistic. Nevertheless, even with the upper bound the reduction quickly grows enormously.
\end{remark}

\section{Code design criterion for Fuchsian code optimization}
\label{optimization}
In the previous sections, we have shown how to construct Fuchsian codes from scratch and how to decode them with the point reduction algorithm (PRA). However, the simulations we have carried out (cf. \cite{WCC}) demonstrate that the typical design criterion for codes used in conjunction with a ML (or lattice decoder) does not work for the PRA decoder. Hence, there is a call for a new design criterion for codes to be used in conjunction with PRA decoder.

In more detail, for ML decoding, the performance is well dictated by the normalized minimum distance, and hence the goal is to maximize the function
\begin{equation}\label{normMinDist}
\Delta_{ML}(\Cc)=\frac{d^2_{min}(\Cc)}{P_{av}(\Cc)},
\end{equation}
where
$$d^2_{min}(\Cc)=\min_{x,x'\in\Cc}\left\{|x-x'|^2\,|\,  x\neq x'\right\}$$
is the squared minimum distance between distinct codewords, and
\begin{equation}\label{energy}
P_{av}=\frac{1}{|\mathcal{C}|}\sum_{\gamma(\tau)\in\Cc}|\gamma(\tau)|^2
\end{equation}
is the average transmission power  of $\Cc$.

In this section, our aim is to develop a similar function $\Delta_{PRA}$ that predicts the performance of Fuchsian codes with the point reduction algorithm. The key design metric stems from the fact that a decoding error will happen if the noise is so big that the received point belongs to a different tile than the one containing the transmitted point, and hence the point reduction algorithm returns a wrong point. Therefore, it is crucial to choose the fundamental domain and the center of the code in such a way that all the codewords have maximal possible  distance to the \emph{decoding border}, \emph{i.e.}, to the closest isometric circle defining the closest neighboring tile. We will refer to this distance as \emph{border distance}. Let $b_x$ be the closest point on the closest isometry circle to $x$. We define the minimum border distance of a Fuchsian code formally as follows.
\begin{definition}
The \emph{minimum border distance} of a Fuchsian code $\Cc$ is
$$
bd^2_{min}(\Cc)=\min_{x\in\Cc}\left\{|x-b_x|^2\right\}
$$
\end{definition}

We have arrived at the following design criterion.

\begin{framed}\textbf{Code design criterion for Fuchsian codes}\vspace{2pt} \\
In order to optimize the performance of a Fuchsian code $\Cc$ with a point reduction algorithm decoder one should seek to maximize the \emph{normalized minimum border distance} function
$$\Delta_{PRA}(\Cc)=\frac{bd^2_{min}(\Cc)}{P_{av}(\Cc)},$$
where  $P_{av}(\Cc)$ is the average transmission power of $\Cc$.
\end{framed}

Notice that in order to fairly compare the functions of different decoding algorithms, one should compare $\Delta_{ML}$ to $4\Delta_{PRA}$, since $d^2_{min}(\Cc)=4bd^2_{min}(\Cc)$ for symmetric codebooks with symmetric decoding regions. That is, the distance between the points is the distance from the first point to the border, and from the border to the second point, which is two times the border distance, giving the constant 4 due to squaring. This criterion motivates future work on code optimization by considering different Fuchsian groups, tessellations, and centers $\tau$.

\begin{remark}
In \cite{WCC} we have observed that the best of our $4$-NUF codes is outperformed by the $4$-QAM except for very low SNRs. On the other hand,  the gap to the worst 4-NUF is so vast that it gives hope to improve by another similar gap, which would bring us very close to 4-QAM. Considering the logarithmic decoding complexity\footnote{In general, there are also fast decoding algorithms for QAM constellations, but they come with a complexity--performance tradeoff, meaning that also the performance of a QAM constellation is degraded if we use suboptimal algorithms that are faster.}, some performance loss can easily be tolerated.
\end{remark}

\section{Conclusions and further research}
\label{conclusions}
In this paper, we have designed a new class of codes called Fuchsian codes. These codes were obtained by considering constellations on the complex plane arising from the  M\"{o}bius transformation related to a Fuchsian group coming from units in rational quaternion algebras.

We have described the construction  and decoding process of the proposed codes in full detail, providing also numerous explicit examples. According to \cite{WCC}, the differences in the performance of different Fuchsian codes can vary drastically. Hence,  as a motivation for future work, we have provided a design criterion in order to construct optimal Fuchsian codes, after having given the preliminary guidelines and first ad hoc constructions in this paper.

In forthcoming work we will  apply the construction method presented herein to different groups, fundamental domains, tessellations, generators, and centers $\tau$,  hopefully being able to significantly improve the performance. In addition, preliminary studies suggest that the point reduction algorithm can be improved at the penalty of increasing the worst-case complexity order to $O(\log^2|\Cc|)$. A remarkable advantage of our construction is its generality, giving us an enormous design space.

 We will also consider the issue of error correction after point reduction, while not substantially increasing the complexity. Another interesting extension is to consider Fuchsian codes for fading  channels and multi-antenna communications.

\section{Acknowledgments}
The authors gratefully acknowledge the support from the European Science Foundation's \emph{COST Action IC1104} and from the research project MTM2012-33830 (MICINN/UB, Spain), as well as the hospitality of the Institute of Mathematics at the University of Barcelona (IMUB).
They would also like to thank Peter Moss from the British Broadcasting Corporation (BBC) Research \& Development for fruitful discussions on AWGN channels, and Professor Pilar Bayer from University of Barcelona for sharing her extensive knowledge on Fuchsian groups.



\begin{thebibliography}{10}
\providecommand{\url}[1]{{#1}}
\providecommand{\urlprefix}{URL }
\expandafter\ifx\csname urlstyle\endcsname\relax
  \providecommand{\doi}[1]{DOI~\discretionary{}{}{}#1}\else
  \providecommand{\doi}{DOI~\discretionary{}{}{}\begingroup
  \urlstyle{rm}\Url}\fi

\bibitem{alsinabayer}
Alsina, M., Bayer, P.: Quaternion orders, binary forms and Shimura curves,
  \emph{CRM Monograph Series}, vol.~22.
\newblock American Mathematical Society (2004)

\bibitem{bayerremon}
Bayer, P., Rem\'on, D.: A point reduction algorithm for cocompat {Fuchsian}
  groups.
\newblock Revised version submitted  (2013).
\newblock Manuscript available from the second author upon request.

\bibitem{viterbo2}
Belfiore, J., Oggier, F., Viterbo, E.: Cyclic division algebras: a tool for
  space--time coding.
\newblock Foundations and Trends in Communications and Information Theory
  \textbf{4}, 1--195 (2007)

\bibitem{WCC}
Blanco-Chac\'on, I., Rem\'on, D., Hollanti, C.: {Fuchsian} codes for {AWGN}
  channels.
\newblock In: Preproceedings of the International Workshop on Coding and
  Cryptography (WCC 13), pp. 495--507. Bergen, Norway (2013).
\newblock {arxiv.org/abs/1307.7252}

\bibitem{Can84}
Cannon, J.: The combinatorial structure of cocompact discrete hyperbolic
  groups.
\newblock Geometriae Dedicata \textbf{16}(2), 123--148 (1984)

\bibitem{brazilian_COAM}
Carvalho, E., Andrade, A., Palazzo, R., Filho, J.V.: Arithmetic {Fuchsian}
  groups and space--time block codes.
\newblock Comput. Appl. Math. \textbf{30}, 485--498 (2011)

\bibitem{capi}
Corrales, C., Jespers, E., Leal, G., del R\'io, A.: Presentations of the unit
  group of an order in a non-split quaternion algebra.
\newblock Advances in Mathematics \textbf{186}, 498--524 (2004)

\bibitem{IT}
Cover, T.M., Thomas, J.: Elements of Information Theory.
\newblock John Wiley and Sons, Inc. (1991)

\bibitem{duke}
Duke, W., Schulze-Pillot, R.: Representation of integers by positive ternary
  quadratic forms and equidistribution of lattice points on ellipsoids.
\newblock Invent. Math. \textbf{99}, 49--57 (1990)

\bibitem{dvb}
DVB: {Digital Video Broadcasting Project, The Global Standard for Digital
  Television}.
\newblock \url{dvb.org}

\bibitem{elkies}
Elkies, N.: Excellent codes from modular curves.
\newblock Proceedings of the thirty-third annual {ACM} symposium on theory of
  computing pp. 200--208 (2001)

\bibitem{fine}
Fine, B., Rosenberger, G.: Algebraic generalizations of discrete groups,
  \emph{Monographs and textbooks in pure and applied Mathematics}, vol. 223.
\newblock Marcel Dekker (1999)

\bibitem{59paper}
Gertsenshtein, M., Vasilev, V.: Waveguides with random inhomogeneties and
  {Brownian} motion in the {Lobachevsky} plane.
\newblock Theory Probab. Appl. \textbf{4}, 391--398 (1959)

\bibitem{maxorder}
Hollanti, C., Lahtonen, J.: A new tool: Constructing {STBC}s from maximal
  orders in central simple algebras.
\newblock In: IEEE Information Theory Workshop (ITW '06), Punta del Este,
  Uruguay, pp. 322--326 (2006)

\bibitem{katok}
Katok, S.: Fuchsian Groups.
\newblock Chicago Lectures in Mathematics Series. The University of Chicago
  Press (1992)

\bibitem{viterbo1}
Oggier, F., Viterbo, E.: Algebraic number theory and code design for rayleigh
  fading channels.
\newblock Foundations and Trends in Communications and Information Theory
  \textbf{1}, 333--415 (2004)

\bibitem{SRS}
Sethuraman, B.A., Rajan, B., Shashidhar, V.: Full-diversity, high-rate
  space--time block codes from division algebras.
\newblock IEEE Transactions on Information Theory \textbf{49}(10), 2596--2616
  (2003)

\bibitem{shimura1967}
Shimura, G.: Construction of class fields and zeta functions of algebraic
  curves.
\newblock Annals of Math. \textbf{85}, 58--159 (1967)

\bibitem{brazilian_franklin2}
da~Silva, E.B., Firer, M., Costa, S.R., Palazzo, R.: Signal constellations in
  the hyperbolic plane: A proposal for new communication systems.
\newblock Journal of the Franklin Institute \textbf{343}, 69--82 (2006)

\bibitem{simon}
Simon, D.: Solving quadratic equations using unimodular quadratic forms.
\newblock Math. Comp. \textbf{74}, 1531--1534 (2005)

\bibitem{brazilian_franklin}
de~Souza, M., Faria, M.B., Palazzo, R., Firer, M.: Edge-pairing isometries and
  counting dirichlet domains on the densest tessellation (12g-6,3) for signal
  set design.
\newblock Journal of the Franklin Institute \textbf{349}, 1139--1152 (2012)

\bibitem{Tal10}
Talambutsa, A.: Attainability of the minimal exponent of exponential growth for
  some fuchsian groups.
\newblock Mathematical Notes \textbf{88}(1), 144--148 (2010)

\bibitem{brazilian_IEEE}
Vieira, V.L., Palazzo, R., Faria, M.B.: On the arithmetic {Fuchsian} groups
  derived from quaternion orders.
\newblock Proceedings of the International Telecommunications Symposium (ITS
  2006), Fortaleza-Ce (Brazil)  (2006)

\bibitem{vigneras}
Vigneras, M.: Arithm{\'e}tique des alg{\`e}bres de quaternions.
\newblock No. 800 in Lecture Notes in Math. Springer-Verlag (1980)

\bibitem{ViBu}
Viterbo, E., Boutros, J.: A universal lattice code decoder for fading channels.
\newblock IEEE Transactions on Information Theory \textbf{45} (1999)

\bibitem{voi09}
Voight, J.: Computing fundamental domains for {Fuchsian} groups.
\newblock Journal de th{\'e}orie des nombres de Bordeaux \textbf{21}, 467--489
  (2009)

\end{thebibliography}

\newpage\section*{APPENDIX: Generation of the constellations without generators}
\label{appendix}
We address now the problem of how to produce the $4$-tuples $(x,y,z,t)\in\Z^4$ such that $x^2-ay^2-bz^2+abt^2=1$, in the case in which the generators of the Fuchsian group $\Gamma(D,1)$ are not known or they are too complex to determine. This construction will be used to obtain the matrices of $\Gamma(D,1)$ acting on $\tau$ by M\"{o}bius transforms. In the first subsection, we develop a constructive method for an infinite family of quaternion algebras, the so called small ramified quaternion $\mathbb{Q}$-algebras of type A \cite{alsinabayer}, while in the second subsection, we show that an infinite subset of elements of the constellation can be produced in general provided that we are able to solve the attached normic equation.

\subsection*{Explicit constructive method}
We will suppose here that our quaternion algebras have the form $\left(\frac{p,-1}{\mathbb{Q}}\right)$ with $p>0$ prime and $p\equiv 3\pmod{4}$. This case is known in the literature as small ramified of type A, and the discriminant of the quaternion algebra in this case is $2p$.

Let us write the quaternion matrix $\gamma=\gamma(x,y,z,t)$ (cf. Section \ref{algebraic}) in terms of the four integer symbols $(x,y,z,t)$ involved. Notice that only three symbols in each $4$-tuple are independent, hence, we would like to parametrize the set of these $4$-tuples by an infinite set of $3$-tuples $(m,k_1,k_2)\in\Z$. Since the quaternion algebra $\left(\frac{p,-1}{\Q}\right)$ is indefinite, one has that the normic equation $x^2-py^2+z^2-pt^2=1$ has infinitely many integer solutions (cf.\cite{alsinabayer}). It is possible to parametrize all the rational solutions of this normic equation by means of rational functions in three variables, but using this method to produce integer solutions seems a difficult task. Instead, we will develop an explicit method to produce an infinite set of such solutions in the small ramified type A. Next, we describe our construction in detail.

First, notice that for $p\equiv 3\pmod{4}$, the ring of integers of the number field $\Q(\sqrt{p})$ is $\Z[\sqrt{p}]$. The multiplicative group of units of this ring is $\left\{\pm\varepsilon^m:m\in\Z\right\}$, where $\varepsilon$ is a unit of infinite order (called a fundamental unit). This is a very particular version of Dirichlet's theorem on units. We have provided the fundamental units in Table \ref{table1} in order to make our method implementable in general. In most symbolic algebra packages like Sage or Magma, it is easy to obtain extensive lists of fundamental units.

Given an element $\theta=a+\sqrt{p}b\in \Q(\sqrt{p})$, let us denote by $\theta'$ its Galois conjugate, \emph{i.e.}, $a-\sqrt{p}b$. For the rest of this section, we will denote by $\varepsilon$ a fundamental unit of $\Z[\sqrt{p}]$ and will suppose $\varepsilon>0$, by taking a Galois conjugate and multiplying by $-1$ if necessary.

Given  a triple $(m,k_1,k_2)$ of nonnegative integers ($m\neq 0$), define $a_m+\sqrt{p}b_m=\varepsilon^m$. We have that $a_m^2-pb_m^2=\varepsilon^m(\varepsilon')^m=1$. Now, set $x_{m,k_1}+\sqrt{p}y_{m,k_1}:=a_m\varepsilon^{k_1}$ and $z_{m,k_2}+\sqrt{p}t_{m,k_2}:=\sqrt{p}b_m\varepsilon^{k_2}$. Notice that $x_{m,k_1}^2-py_{m,k_1}^2=a_m^2$ and $z_{m,k_2}^2+pt_{m,k_2}^2=-pb_m^2$, hence
$$x_{m,k_1}^2-py_{m,k_1}^2+z_{m,k_2}^2-pt_{m,k_2}^2=a_m^2-pb_m^2=1.$$
We will use the notation
$$
\phi_p(m,k_1,k_2)=(x_{m,k_1},y_{m,k_1},z_{m,k_2},t_{m,k_2}),
$$
making it evident that we can parametrize an infinite subset of integer points of the hyper quadric $x^2-py^2+z^2-pt^2=1$ by using three variables .

\begin{proposition}For each prime number $p\equiv 3\pmod{4}$, the map $\phi_p$ is bijective over its image, which is contained in the set
$$
\{(x,y,z,t)\in\Z_{\geq 0}^4\, |\, x^2-py^2=m^2, z^2-pt^2=-pr^2,\mbox{ for some }m,r\in\Z\}.
$$
\end{proposition}
\begin{proof}Let $(m_1,k_{1,1},k_{1,2})$ and $(m_2,k_{2,1},k_{2,2})$ be two triples of nonnegative integers with $m_1,m_2\neq 0$. Suppose $m_1=m_2=m$. If $k_{1,1}\neq k_{2,1}$, then $a_m\varepsilon^{k_{1,1}}\neq a_m\varepsilon^{k_{2,1}}$ and $\phi_p(m_1,k_{1,1},k_{1,2})\neq \phi_p(m_2,k_{2,1},k_{2,2})$. The case  $k_{2,1}\neq k_{2,2}$ is analogous. Suppose that $m_1\neq m_2$. In this case, $a_{m_1}\neq a_{m_2}$ or $b_{m_1}\neq b_{m_2}$. Suppose that $a_{m_1}\neq a_{m_2}$. In this case, $a_{m_1}\varepsilon^{k_{1,1}}\neq a_{m_2}\varepsilon^{k_{2,1}}$, since otherwise, $a_{m_1}^2=a_{m_2}^2$, and since $\varepsilon>0$, we would have that $a_{m_1}=a_{m_2}$. The remaining case is identical.
\end{proof}
As an illustration of our method, in Table \ref{table2} the reader can find a set of images of the parametrization $\phi_p$ for some different values of $p$ and different domain entries.

\begin{table}
\caption{Fundamental units $\varepsilon$ for $\Z[\sqrt{p}]$, $p\equiv 3\pmod{4}$, $p<50$}
\label{table1}
\begin{tabular}{|c|c|| c|c|}
\hline $p=3$ & $2+\sqrt{3}$  & $p=23$ & $24+2\sqrt{23}$  \\
\hline $p=7$ & $8+3\sqrt{7}$  & $p=31$ & $1520+237\sqrt{31}$ \\
\hline $p=11$ & $10+3\sqrt{11}$ & $p=43$ & $3482+531\sqrt{43}$ \\
\hline $p=19$ & $170+39\sqrt{19}$ & $p=47$ & $48+7\sqrt{47}$ \\
\hline
\end{tabular}%
\end{table}

\begin{table}[H]
\caption{Explicit parametrization at different values for $p=3,7,11$}
\label{table2}
\begin{tabular}{|c|c|c|c|}
\hline $(m,k_1,k_2)$ & $\phi_3$ & $\phi_7$ & $\phi_{11}$ \\
\hline $(1,0,1)$ & $(2,0,3,2)$ & $(8,0,63,24)$ & $(10,0,99,30)$\\
\hline $(2,0,1)$ & $(7,0,12,8)$ & $(127,0,1008,384)$ & $(199,0,1980,600)$\\
\hline $(2,1,1)$ & $(14,7,12,8)$ & $(1016,381,1008,384)$ & $(1990,597,1980,600)$\\
\hline
\end{tabular}%
\end{table}

The explicit parametrization of the whole group of units is a delicate problem. On the contrary to the number field setting, the structure of the group of units of reduced norm $1$ in quaternion algebras has not been explicitly described yet. However, there exist some interesting theoretical results, see \cite{capi}.

\subsection*{General theoretical approach}

Let $H=\left(\frac{a,b}{\Q}\right)$ be an indefinite quaternion $\mathbb{Q}$-algebra of discriminant $D_H>1$. Let us recall the following result.

\begin{lemma}\cite[Thm. 4.3]{alsinabayer} Let $H$ be a quaternion $\mathbb{Q}$-algebra with discriminant $D_H$ and $F$ a quadratic number field with discriminant $D_F$. The following statements are equivalent:
\begin{itemize}
\item[1)] There exists an embedding of $F$ into $H$.
\item[2)] For every prime number $q$ such that $q\mid D_H$, $\left(\frac{D_F}{q}\right)\neq 1$.
\end{itemize}
\label{embeddings1}
\end{lemma}

We are interested in quadratic fields $\mathbb{Q}(\sqrt{q})$ which can be embedded in a quaternion algebra of discriminant $D_H$. Again, we assume that $q\equiv 3\pmod{4}$. As an application of the above lemma, we have the following result, which will be used later.

\begin{lemma}For $H=\left(\frac{p_1p_2}{\mathbb{Q}}\right)$, small ramified quaternion $\mathbb{Q}$-algebra, the set $A=\{q\equiv 3\pmod{4},q\mbox{ prime, }\mathbb{Q}(\sqrt{q})\hookrightarrow H\}$ is infinite
\end{lemma}
\begin{proof}
According to lemma \ref{embeddings1}, the primes $q$ such that $\mathbb{Q}(\sqrt{q})$ embeds in $H$, are precisely those such that $\left(\frac{4q}{p_1}\right), \left(\frac{4q}{p_2}\right)\neq 1$. Hence, we are seeking for prime numbers $q\equiv 3\pmod{4}$ such that both Legendre symbols are either $0$ or $-1$. The existence of such primes is granted by the Chinese reminder theorem and Dirichlet's theorem on primes in arithmetic progressions: indeed, the map
$$
\begin{array}{ccc}
\mathbb{Z}/p_1p_2\mathbb{Z} & \to & \mathbb{Z}/p_1\mathbb{Z}\times\mathbb{Z}/p_2\mathbb{Z}\\
([a]_{p_1p_2}) & \mapsto & ([a]_{p_1},[a]_{p_2})
\end{array}
$$
is surjective, hence we can take $([a]_{p_1},[b]_{p_2})\in \mathbb{Z}/p_1\mathbb{Z}\times\mathbb{Z}/p_2\mathbb{Z}$ such that $a$ is a non-square modulo $p_1$ and $b$ a non-square modulo $p_2$ (we have $\frac{(p_1-1)(p_2-1)}{4}$ pairs of this kind) and find an inverse image $x_0$ modulo $p_1p_2$. Now, for $p_1,p_2>2$, we apply again the Chinese reminder reminder theorem to find $x$ in $\mathbb{Z}/4p_1p_2\mathbb{Z}$ congruent to $x_0$ modulo $p_1p_2$ and to $-1$ modulo $4$. Now, by Dirichlet's theorem on primes in arithmetic progressions, we have infinitely many prime numbers in the class of $x$ modulo $4p_1p_2$.
\end{proof}

Fixing an embedding of $\Q(\sqrt{q})$ into $H$ is equivalent to fix  a pure quaternion $\omega=xI+yJ+zK\in H$ of norm $-p$, that is $(x,y,z)\in\Z^3$ such that $ax^2+by^2-abz^2=q$. Since $H$ is indefinite, this normic equation has infinitely many solutions, hence, there exist bijections $\varphi_p:\mathbb{N}\to \{(x,y,z)\in\Z^3:ax^2+by^2-abz^2=q\}$. Determining such a bijection is equivalent to solve the diophantine equation $ax^2+by^2-abz^2=q$, which is a classical problem in number theory. It is possible to give asymptotic estimates of the number of solutions, which involves the use of modular forms of fractional weight $3/2$ (cf. \cite{duke}). Nevertheless, there exists a polynomial algorithm which computes finite sets of solutions (cf. \cite{simon}).

Now, given a real quadratic field $\Q(\sqrt{q})$ embedded in $H$ we can obtain units in the natural order $\mathbb{Z}[1,I,J,K]$ of $H$ from the group of units of the ring of integers of the quadratic field, generated by $\varepsilon=x+y\sqrt{q}$  (notice that the fundamental unit $\varepsilon$ is usually normalized so that $x,y>0$ and its absolute value is greater than $1$, by taking Galois conjugate and/or changing sign, if necessary). Thus, identifying the units in the quaternion order with the corresponding matrices in the arithmetic Fuchsian group $\Gamma(D,1)$, we define maps $\psi_q:\mathbb{N}^2\to \Gamma(D,1)$ given by $\psi_q(t,m)=\left(x+y\varphi_q(t)\right)^m$.

\begin{proposition}The map $\psi_q$ is injective when restricted to $\mathbb{N}\times\left(\mathbb{N}\setminus\{0\}\right)$.
\end{proposition}

\begin{proof}
Consider $\psi_q(t_1,m_1)=\psi_q(t_2,m_2)$. Suppose first that $m_1=m_2=m$. Then, since we have $\varepsilon^m=l+r\sqrt{q}$, with $r\neq 0$, from $
l+r\varphi_q(t_1)=l+r\varphi_q(t_2)
$, we deduce  $\varphi_q(t_1)=\varphi_q(t_2)$, hence $t_1=t_2$.

Suppose now that $m_1\neq m_2$. In this case, setting $\psi_q(t_1,n_1)=l_1+m_1\varphi_q(t_1)$ and $\psi_q(t_2,m_2)=l_2+r_2\varphi_q(t_2)$, since $r_i\varphi_p(t_i)$ is a pure quaternion, we have that $l_1=l_2$.
However, by the binomial formula, and taking into account that $\varphi_q(t_1)^2=\varphi_q(t_2)^2=q$, we have
$$
l_1=\sum_{\tiny\begin{array}{l}j=0 \\ j \text{ even} \end{array}}^{m_1}{{m_1}\choose{j}}y^{j}q^{\frac{j}{2}}x^{m_1-j}.$$
Now assume  $m_1>m_2\geq j$, so we have that ${{m_1}\choose{j}}>{{m_2}\choose{j}}$, hence
$$ l_1 > \sum_{\tiny\begin{array}{l}j=0 \\ j \text{ even} \end{array}}^{m_2}{{m_1}\choose{j}}y^{j}q^{\frac{j}{2}}x^{m_1-j}
>\sum_{\tiny\begin{array}{l}j=0 \\ j \text{ even} \end{array}}^{m_2}{{m_2}\choose{j}}y^{j}q^{\frac{j}{2}}x^{m_2-j}=l_2,
$$
which is a contradiction.
\end{proof}

With these maps we can produce a countable family of non-overlapping infinite families of codewords:
\begin{proposition}Let $q_1,q_2\equiv 3\pmod{4}$ be two different prime numbers such that
$\Q(\sqrt{p_1}),\Q(\sqrt{q_2})\hookrightarrow H$. Then
$\psi_{q_1}(t_1,m_1)=\psi_{q_2}(t_2,m_2)$ if and only if $m_1=m_2=0$.
\end{proposition}
\begin{proof}The \emph{if} clause is trivial.
Suppose $\psi_{q_1}(t_1,m_1)=\psi_{q_2}(t_2,m_2)$. Writing $\psi_{q_1}(t_1,m_1)=l_1+r_1\varphi_{q_1}(t_1)$ and $\psi_{q_2}(t_2,m_2)=l_2+r_2\varphi_{q_2}(t_2)$, we have that $l_1=l_2$ and $r_1\varphi_{q_1}(t_1)=r_2\varphi_{q_2}(t_2)$. Taking squares we obtain $r_1^2q_1=r_2^2q_2$, which implies $r_1=r_2=0$ and, since $q_1,q_2>0$, we deduce that $m_1=m_2=0$.
\end{proof}

The above facts, allow us to conclude the following

\begin{theorem}Let $H=\left(\frac{p_1,p_2}{\mathbb{Q}}\right)$ be a small ramified quaternion $\mathbb{Q}$-algebra of discriminant $D$ with $p_1\equiv 3\pmod{4}$ square free. Let $\Gamma$ be the subgroup of $\Gamma(D,1)$ consisting of matrices with entries in $\Z[\sqrt{p_1}]$. There exists a parametrization of an infinite subset of $\Gamma$  by three degrees of freedom.
\end{theorem}
\begin{proof}Let $A$ be the infinite set of primes $\equiv 3\pmod{4}$ such that $\Q(\sqrt{q})$ embeds into $H$. For any $p\in A$, fix a generator of the unit group of the form $x_q+y_q\sqrt{q}$ with $x_q,y_q>0$. Now, the map $\Psi:A\times\mathbb{N}\times\left(\mathbb{N}\setminus\{0\}\right)\to\Gamma(D,1)$ defined by $\Psi(q,s,m)=\psi_q(s,m)$ is injective. 
\end{proof}

\begin{remark} Notice that this theorem is not explicit, since it depends on how to produce the solutions of the normic form. But using the algorithm described in \cite{simon}, we can explicitly parametrize an infinite family of units by two dregrees of freedom. Further studies on the structure of the group of units will allow us to make the full parametrization more explicit.
\end{remark}

\subsection*{Relation to the size duplication}


 If a matrix $\gamma$ corresponds to the $4$-tuple $(x,y,z,t)$, and this $4$-tuple corresponds to the $3$-tuple $(m,k_1,k_2)$ of independent nonnegative integers, then the matrix $-\gamma$ corresponds to the $3$-tuple $(-m,k_1,k_2)$. Notice that this is not ambiguous since the original triples are assumed to have nonnegative entries, and $\theta>0$.

To recover the right $3$-tuple from a received signal, we  first check whether it belongs to $\mathcal{H}$ or to $-\mathcal{H}$. In the first case, we use the point reduction algorithm to obtain $(x,y,z,t)$ and the parametrization to obtain $(m,k_1,k_2)$. In the second case, we have received $v=-\gamma_k(\tau)+n$, hence, we apply the point reduction algorithm to $-v$, obtain $(x,y,z,t)$ and $(m,k_1,k_2)$, and we decode it as $(-m,k_1,k_2)$.

\end{document}